\documentclass[11pt]{article}
\def\withcolors{0}
\def\withnotes{0}
\usepackage[T1]{fontenc}
\usepackage[utf8]{inputenc}

\usepackage{lmodern}
\usepackage{xspace}                                     

\usepackage[normalem]{ulem}

\usepackage{amsfonts,amsmath,amssymb, amsthm, mathtools}
\usepackage{thm-restate}
\usepackage{dsfont} 

\usepackage{algorithmicx,algpseudocode,algorithm}

\usepackage[usenames,dvipsnames,table]{xcolor}

\usepackage{csquotes}

\usepackage{relsize}

\usepackage{multirow}
\usepackage{chngpage} 

\usepackage{mfirstuc}

\usepackage{tikz}
\usetikzlibrary{arrows}
\usetikzlibrary{calc,decorations.pathmorphing,patterns}

\usepackage[backref,colorlinks,citecolor=blue,bookmarks=true,linktocpage]{hyperref}
\usepackage{aliascnt}
\usepackage[numbered]{bookmark}

\usepackage{fullpage}

\usepackage{titling}

\usepackage[shortlabels]{enumitem}
  \setitemize{noitemsep,topsep=3pt,parsep=2pt,partopsep=2pt} 
  \setenumerate{itemsep=1pt,topsep=2pt,parsep=2pt,partopsep=2pt}
  \setdescription{itemsep=1pt}
  
\ifnum\withnotes=1
  \usepackage[colorinlistoftodos,textsize=scriptsize]{todonotes}
\fi

\usepackage{verbatim}
\usepackage{bigfoot} 

\usepackage{mleftright} 

\usepackage{pgfplots}

\input{preamble}

\def\authornamecc{Cl\'ement L. Canonne}
\def\authorafficc{Columbia University. Email: \email{ccanonne@cs.columbia.edu}. Research supported by NSF grants CCF-1115703 and NSF CCF-1319788.}
\def\authornametg{Tom Gur}
\def\authoraffitg{Weizmann Institute. Email: \email{tom.gur@weizmann.ac.il}.
Research partially supported by ISF grant 671/13.
}

\title{An Adaptivity Hierarchy Theorem for Property Testing\thanks{This work previously appeared as ``Fifty Shades of Adaptivity (in Property Testing).''}}
\date{\today}
 
\author{
  \ccolor{\authornamecc}\thanks{\authorafficc}
  \and \tcolor{\authornametg}\thanks{\authoraffitg}
}

\begin{document}

\maketitle

\pagenumbering{gobble}
\begin{abstract}
Adaptivity is known to play a crucial role in property testing. In particular, there exist properties for which there is an exponential gap between the power of \emph{adaptive} testing algorithms, wherein each query may be determined by the answers received to prior queries, and their \emph{non-adaptive} counterparts, in which all queries are independent of answers obtained from previous queries.

In this work, we investigate the role of adaptivity in property testing at a finer level. We first quantify the degree of adaptivity of a testing algorithm by considering the number of ``rounds of adaptivity'' it uses. More accurately, we say that a tester is $k$-(round) adaptive if it makes queries in $k+1$ rounds, where the queries in the $i$'th round may depend on the answers obtained in the previous $i-1$ rounds. Then, we ask the following question:

\begin{quote}\itshape
{Does the power of testing algorithms smoothly grow with the number of rounds of adaptivity?}
\end{quote}

\noindent We provide a positive answer to the foregoing question by proving an adaptivity hierarchy theorem for property testing. Specifically, our main result shows that for every $n\in \N$ and $0 \le k \le n^{0.99}$ there exists a property $\property_{n,k}$ of functions for which (1) there exists a $k$-adaptive tester for $\property_{n,k}$ with query complexity $\tildeO{k}$, yet (2) any $(k-1)$-adaptive tester for $\property_{n,k}$ must make $\Omega(n)$ queries. In addition, we show that such a qualitative adaptivity hierarchy can be witnessed for testing natural properties of graphs.
\end{abstract}

\ifnum\withnotes=0
  \clearpage
  \tableofcontents
\else
  \clearpage
  \listoftodos
  \hfill
  \tableofcontents
  \section*{To Consider}
  \begin{enumerate}
  		\item Testing signed majorities.
     \item k-juntas or k-parities: natural hierarchy?
  		\item MAP/IPP adaptivity hierarchy.
  		\item Distribution testing (COND).
  \end{enumerate}

  \clearpage
\fi

\clearpage\pagenumbering{arabic}


\section{Introduction}
The study of property testing, initiated by Rubinfeld and Sudan~\cite{RS:96} and Goldreich, Goldwasser and Ron~\cite{GGR:98}, has attracted significant attention in the last two decades (see, e.g., recent books~\cite{Gol:10,Gol:17,BY:17} and surveys~\cite{Ron:08,Ron:09,Canonne:15:Survey}). Loosely speaking, property testers are highly efficient randomized algorithms (typically running in sublinear time) that solve approximate decision problems, while only inspecting a tiny fraction of their inputs. More accurately, an $\eps$-tester $\Tester$ for property $\property$ is a randomized algorithm that, given query access to an input $x$, decides whether $x \in \property$ or $x$ is $\eps$-far (say, in Hamming distance) from $\property$. The query complexity of $\Tester$ is then the number of queries it makes to $x$.

In general, a testing algorithm may select its queries adaptively such that the $i$'th query is determined by the answers to the previous $i-1$ queries, in which case it is said to be an \emph{adaptive tester}. However, in many natural cases, testers may actually determine their queries solely based on their randomness (and input length), without any dependency on answers to previous queries; a tester that satisfies this condition is called a \emph{non-adaptive tester}. A natural question, which commonly arises in query-based models,  is whether the ability to make adaptive queries can significantly affect the query complexity. 

Adaptive queries can be easily emulated at the cost of a large blowup in query complexity (exponential in the number of queries). More accurately, any $q$-query adaptive tester for a property of objects represented by functions $f\colon D \to R$ can be emulated by an $|R|^q$-query non-adaptive tester (see e.g.,~\cite[Section 1.5]{Gol:17}). While for certain types of properties and models~--~e.g., linear properties~\cite{BHR:05} and properties in the dense graph model~\cite{GT03}~--~one has better emulations which come with little or no overhead, such efficient emulations cannot exist for all properties. As was shown by Raskhodnikova and Smith~\cite{RS:06}, in the bounded-degree graph model~\cite{GRexp:00} there is a large chasm between the adaptive and non-adaptive query complexities of testing many natural graph properties. In particular, any property over bounded-degree graphs with $n$ vertices, which is not determined by the \emph{vertex degree distribution},\footnote{Loosely speaking, a property $\property$ of bounded-degree graphs is not determined by the vertex degree distribution if there exist two graphs, $G_1 \in \property$ and $G_2$ that is ``far'' from $\property$, such that the vertices of $G_1$ and $G_2$ have the same degrees.} requires $\Omega(\sqrt{n})$ queries to test non-adaptively, whereas many such properties (e.g., triangle-freeness and connectivity) have $\eps$-testers with query complexity $\poly(1/\eps)$.

In this work, we investigate the role of adaptivity in property testing at a finer level. Rather than considering the extreme cases of fully adaptive testers versus completely non-adaptive testers, we consider testers with various levels of restricted adaptivity and ask the following question:

\begin{quote}\itshape
Can the power of testers gradually grow with the ``amount'' of adaptivity they are allowed to use?
\end{quote}

\noindent Besides the sheer theoretical interest of understanding the role of adaptivity in property testing, a motivation for this question comes from the \emph{constraints} that come with adaptive algorithms, which may counterbalance the apparent gain in efficiency. Indeed, non-adaptive algorithms (or at least those which only use a small number of adaptive ``stages'') may be preferred in practice to their adaptive counterparts, in spite of the larger number of queries they make. The reason for this preference is the significant gains obtained by being able to make many queries \emph{in parallel}: when each query is an experiment which, while relatively cheap by itself, may take several hours, assessing the trade-off between rounds of adaptivity and total number of queries becomes crucial. An archetypal example where such considerations prevail is the (different) setting of group testing (see e.g.~\cite[Section 1.2]{Du:GT:2000}).

To answer the foregoing question, we shall first need to give a precise definition for the ``amount'' of adaptivity that a tester uses. To this end, it is natural to consider the number of ``rounds of adaptivity'' used by a tester.\footnote{We also consider an alternative notion of \emph{tail adaptivity}, which roughly speaking refers to testers that first make a large number of non-adaptive queries and subsequently make a bounded number of adaptive queries. See~\autoref{sec:def} for details regarding how these two notions relate.} More precisely, we say that a tester is \emph{$k$-round-adaptive} if it generates and makes queries in $k+1$ rounds, where in the $i$'th round the tester queries a set of locations $Q_i$ that may depend on the answers to queries in $Q_0,\ldots,Q_{i-1}$, obtained in previous rounds. We will quantify the  ``amount'' of adaptivity that a tester uses by the number of rounds of adaptivity that it uses. Equipped with the notion of round adaptivity, we can proceed to present our results.

\subsection{Our Results}\label{ssec:our:results}
Our main result provides a positive answer to the foregoing question by showing an adaptivity hierarchy theorem for property testing; that is, we show a family of properties $\{\property_k\}_k$ such that for every $k$, the property $\property_k$ is ``easy'' for $k$-adaptive testers and ``hard'' for $(k-1)$-adaptive testers.

\begin{theorem}[Informally stated (see~\autoref{theo:hierarchy:theorem})]
\label{thm:informal1}
	For every $n\in \N$ and $0 \le k \le n^{0.99}$ there is a property $\property_{n,k}$ of strings over $\field{n}$ such that:
	\begin{enumerate}
		\item there exists a $k$-round-adaptive tester for $\property_{n,k}$ with query complexity $\tildeO{k}$, yet
		\item  any $(k-1)$-round-adaptive tester for $\property_{n,k}$ must make $\Omega(n)$ queries.	
	\end{enumerate}
\end{theorem}

The above theorem relies on an arguably contrived family of property, which was specifically tailored towards maximizing the separations; hence, one may wonder whether such strong separations also hold for more natural properties. As we show below, this is indeed the case: namely, we establish another adaptivity hierarchy theorem that, albeit weaker than~\autoref{thm:informal1}, applies to the well-studied natural problem of testing $k$-cycle freeness in the bounded-degree graph model (see~\autoref{sec:cycle_prelim} for definitions).

\begin{theorem}
\label{thm:informal2}
	Let $k\in \N$ be a constant. Then,
	 \begin{enumerate}[(i)]
      \item there exists a $k$-round-adaptive tester with query complexity $O(1/\eps)$ for $(2k+1)$-cycle freeness in the bounded-degree graph model; yet
      \item any $(k-1)$-round-adaptive tester for $(2k+1)$-cycle freeness in the bounded-degree graph model must make $\bigOmega{\sqrt{n}}$ queries, where $n$ is the number of vertices in the graph.
      \end{enumerate}
\end{theorem}

We conclude this section by posing two open problems that naturally arise from our work.

\begin{openquestion}[One property to rule them all]\label{openquestion:single:property}
  Does there exist an adaptivity hierarchy with respect to a \emph{single} property? That is, for any $m$ and all sufficiently large $n$, is there a property $\property$ of elements of size $n$, and $q_1>\ldots>q_m$ ($m$ ``levels'' of hierarchy) such that for every $k\in[m]$ there exists a $k$-adaptive tester for $\property$ with query complexity $q_k$, yet every $(k-1)$-adaptive tester must make $\omega(q_k)$ queries to test $\property$?
\end{openquestion}

\begin{openquestion}[Au naturel is just as good]\label{openquestion:natural:property}
  Does there exist a family of \emph{natural} properties which exhibits an adaptivity hierarchy with separations as strong as in~\autoref{thm:informal1}?
\end{openquestion}



\subsection{Previous Work}\label{ssec:previous:work}
As previously mentioned, the role of adaptivity in property testing has been the focus of several works before. It is well known that for any property of Boolean functions, there exists at most an exponential gap between adaptive and non-adaptive testers: any (adaptive) $q$-query testing algorithm for a property $\property$ of $n$-variate Boolean functions can be simulated by a non-adaptive tester with query complexity $2^q-1$. Further, such gaps are known to exist for some natural properties, such as read-once width-2 OBDDs~\cite{RT:12,BMW:11} and signed majorities~\cite{MORS:09,RS:13} (importantly, there also exist cases where adaptivity is known \emph{not} to help~\cite{BLR:93,BHR:05}). Another prominent example of a class of Boolean functions where adaptivity is known to help is that of $k$-juntas~\cite{Blais:09,Blais:08,ServedioTW:15,CSTWX:17}, which can be tested adaptively with $\tildeO{k}$ queries, yet for which the non-adaptive query complexity is $\tildeTheta{k^{3/2}}$.

Of course, the Boolean function setting is not the only one: in the dense graph model, it is known that while adaptivity can help~\cite{GR:11}, it will be at most by a quadratic factor~\cite{AFKS00,GT03}: that is, every graph property testable (adaptively) with $q$ queries has an $\bigO{q^2}$-query non-adaptive tester. This is no longer the case in the bounded-degree model, however; where Raskhodnikova and Smith showed that there exist many properties which can be tested adaptively with a constant number of queries, but for which any non-adaptive tester must have query complexity $\bigOmega{\sqrt{n}}$~\cite{RS:06}.

However, all these results, even when they establish cases where adaptivity does help, leave open the question of \emph{how much} adaptivity is needed for this to happen.
 In particular, for the case of properties of Boolean functions, many known adaptive testers which outperforms their non-adaptive counterpart do so, at some level, by conducting a binary search of some sort (see, e.g.,~\cite{Blais:09,RT:12,RS:13}) and thus comes inherently with a logarithmic numbers of ``adaptive rounds.'' 
 
Our proof of~\autoref{thm:informal1} relies on a connection between the property testing and linear decision tree models. Although many of the ingredients we use are new, the connection itself is not and was first observed in~\cite{Tell:14} (see also~\cite{BCK:14} for a slightly different connection between property testing and parity decision trees).

\paragraph{Adaptivity in other settings.} We remark that the notion of round complexity in communication complexity and interactive proof systems is somewhat analogous to that of round adaptivity, since in those models each round of communication or interaction allows the parties to adapt their strategies. Moreover, a round complexity hierarchy is known for communication complexity~\cite{NW:93} and interactive proofs of proximity~\cite{GR:17}. Finally, we also mention that the role of the number of adaptive measurements used by sparse recovery algorithms was shown to be very significant~\cite{IPW:11}.

\subsection*{Organization}
In~\autoref{sec:prelims} we provide the preliminaries required for the technical sections. In~\autoref{sec:def} we provide a precise definition for testers with bounded adaptivity. In~\autoref{sec:hierarchy} we prove our main result, which is a strong adaptivity hierarchy theorem for a property of functions. In~\autoref{sec:hierarchy:natural} we prove an adaptivity hierarchy theorem with respect to a natural property of graphs. Finally, in~\autoref{sec:misc} we discuss adaptivity round reductions, as well as a connection to communication complexity, and the relation between round and tail adaptivity.

\section{Preliminaries}
\label{sec:prelims}
We begin with standard notations:
\begin{itemize}
\item We denote the \emph{relative Hamming distance}, over alphabet $\Sigma$, between two vectors $x \in \Sigma^n$ and $y \in \Sigma^n$ by $\dist{x}{y} \eqdef \left| \left\{ x_i \neq y_i \;\colon\; i \in [n] \right\} \right|/n$. If $\dist{x}{y} \leq \eps$, we say that $x$ is \textsf{$\eps$-close} to $y$, and otherwise we say that $x$ is \textsf{$\eps$-far} from $y$. Similarly, we denote the \emph{relative distance} of $x$ from a non-empty set $S \subseteq \Sigma^n$ by $\dist{x}{S} \eqdef \min_{y \in S} \dist{x}{y}$. If $\dist{x}{S} \leq \eps$, we say that $x$ is \textsf{$\eps$-close} to $S$, and otherwise we say that $x$ is \textsf{$\eps$-far} from $S$. 

\item We denote by $A^x(y)$ the output of algorithm $A$ given direct access to input $y$ and oracle access to string $x$. Given two interactive machines $A$ and $B$, we denote by $(A^x,B(y))(z)$ the output of $A$ when interacting with $B$, where $A$ (respectively, $B$) is given oracle access to $x$ (respectively, direct access to $y$) and both parties have direct access to $z$. Throughout this work, probabilistic expressions that involve a randomized algorithm $A$ are taken over the inner randomness of $A$ (e.g., when we write $\Pr[A^x(y) = z]$, the probability is taken over the coin tosses of $A$).

\item We use the notations $\tildeO{f},\tildeOmega{f}$ to hide polylogarithmic dependencies on the argument, i.e. for expressions of the form $\bigO{f \log^c f}$ and $\bigOmega{f \log^c f}$ (for some absolute constant $c$). Finally, all our logarithms are in base $2$.
\end{itemize}

\paragraph{Integrality.} For simplicity of notation, we hereafter use the convention that all (relevant) integer parameters that are stated as real numbers are implicitly rounded to the closest integer.

\paragraph{Uniformity.} To facilitate notation, throughout this work we define all algorithms \emph{non-uniformly}; that is, we fix an integer $n \in \N$ and restrict the algorithms to inputs of length $n$. Despite fixing $n$, we view it as a generic parameter and allow ourselves to write asymptotic expressions such as $O(n)$. We remark that while our results are proved in terms of non-uniform algorithms, they can be extended to the uniform setting in a straightforward manner.

\section{The Definition of Testers with Bounded Adaptivity}
\label{sec:def}
In this section, we provide a formal abstraction that captures the notion of \emph{bounded adaptivity} within the framework of property testing. We define two notions of bounded adaptivity: (1)~\emph{round-adaptivity}, which refers to algorithms that are allowed to make a bounded number of ``batches'' of queries, where the queries in each batch may depend on the answers to previous batches; (2)~\emph{tail-adaptivity}, which refers to algorithms that first make a large number of non-adaptive queries and subsequently make a bounded number of adaptive queries.

We remark that while tail-adaptivity can be easily emulated via round-adaptivity, the converse does \emph{not} hold. Indeed, in~\autoref{sec:misc:round:tail} we show that round-adaptive testers can be much more powerful than tail-adaptive testers. Nonetheless, our lower bounds hold for the stronger round-adaptivity notion, whereas out upper bounds hold for the more restrictive tail-adaptivity.

\begin{definition}[Round-Adaptive Testing Algorithms]\label{def:ra:pt}
Let $\domain$ be a domain of cardinality $n$, and let $k,q \le n$. A randomized algorithm is said to be a \emph{$(k,q)$-round-adaptive} tester for a property $\property\subseteq 2^{\domain}$, if, on proximity parameter $\eps\in(0,1]$ and granted query access to a function $f\colon \domain\to\bool$, the following holds. 
  \begin{enumerate}[(i)]
        \item \textsf{Query Generation}: The algorithm proceeds in $k+1$ rounds, such that at round $\ell\geq 0$, it produces a set of queries $Q_\ell\eqdef\{x^{(\ell),1},\dots,x^{(\ell),\abs{Q_\ell}}\}\subseteq \domain$ (possibly empty), based on its own internal randomness and the answers to the previous sets of queries $Q_0,\dots, Q_{\ell-1}$, and receives $f(Q_\ell)=\{f(x^{(\ell),1}),\dots,f(x^{(\ell),\abs{Q_\ell}})\}$;
    \item \textsf{Completeness}: If $f\in\property$, then the algorithm outputs $\accept$ with probability at least $2/3$;
    \item \textsf{Soundness}: If $\dist{f}{\property} > \eps$, then the algorithm outputs $\reject$ with probability at least $2/3$.
  \end{enumerate}
  The \emph{query complexity} $q$ of the tester is the total number of queries made to $f$, i.e., $q = \sum_{\ell=0}^{k} \abs{Q_\ell}$. If the algorithm returns \accept with probability one whenever $f\in\property$, it is said to have \emph{one-sided} error (otherwise, it has \emph{two-sided} error). We will sometimes refer to a tester with respect to proximity parameter $\eps$ as an $\eps$-tester.
\end{definition}

\begin{remark}[On amplification]\label{rk:repetition}
We note that, as usual in property testing, the probability of success can be amplified by repetition to any $1-\delta$, at the price of an $\bigO{\log(1/\delta)}$ factor in the query complexity. Crucially, this can be done with no increase in the number of adaptive rounds: while repetition would na\"ively multiply both $q$ \emph{and} $k$ by this factor, one can avoid the latter by running the $\bigO{\log(1/\delta)}$ independent copies of the algorithm in parallel, instead of sequentially.
\end{remark}

\begin{definition}[Tail-Adaptive Testing Algorithms]\label{def:ta:pt}
Let $\domain$ be a domain of cardinality $n$, and let $k,q \le n$. A randomized algorithm is said to be a \emph{$(k,q)$-tail-adaptive} tester for a property $\property\subseteq 2^{\domain}$, if, on proximity parameter $\eps\in(0,1]$, error parameter $\delta\in(0,1]$, and granted query access to a function $f\colon \domain\to\bool$, the following holds. 
  \begin{enumerate}[(i)]
    \item \textsf{Query Generation}: The algorithm proceeds in $k+1$ rounds, such that in the first round, it produces a set of queries $Q\eqdef\{x^{(0),1},\dots,x^{(0),\abs{Q}}\}\subseteq \domain$ (possibly empty), based on its own internal randomness; and receives $f(Q)=\{f(x^{(0),1}),\dots,f(x^{(0),\abs{Q}})\}$; then it makes, over the next $k$ rounds, $k$ adaptive queries to $f$, denoted $x^{(1)},\dots,x^{(k)}$;
    \item \textsf{Completeness}: If $f\in\property$, then the algorithm outputs $\accept$ with probability at least $1-\delta$;
    \item \textsf{Soundness}: If $\dist{f}{\property} > \eps$, then the algorithm outputs $\reject$ with probability at least $1-\delta$.
  \end{enumerate}
The \emph{query complexity} $q$ of the tester is the total number of queries made to $f$, i.e., $q =\abs{Q} + k$. If  the algorithm returns \accept with probability one whenever $f\in\property$, it is said to be \emph{one-sided} (otherwise, it is \emph{two-sided}).
\end{definition}

\begin{remark}[On (lack of) amplification]\label{rk:repetition:tail:adaptive}
Unlike the round-adaptive algorithms, tail-adaptive testing algorithms do not enjoy a simple success amplification procedure which would leave unchanged the adaptivity parameter, only affecting the query complexity. This is the reason why the success probability $\delta$ is explicitly mentioned in~\autoref{def:ta:pt}.
\end{remark}

\section{A Strong Adaptivity Hierarchy}\label{sec:hierarchy}
In this section we prove the adaptivity hierarchy theorem, which shows that, loosely speaking, up to a nearly linear threshold, each additional round of adaptivity can significantly augment the power of testing algorithms.

\begin{theorem}[Adaptivity Hierarchy Theorem]
\label{theo:hierarchy:theorem}
  Fix any $\alpha \in (0,1)$. There exists a constant $\beta \in (0,1)$ such that, for every $n\in\N$, the following holds. 
  For every integer $0\leq k\leq n^{\beta}$, there exists a property $\property_k \subseteq \field{n}^{n^{1+\alpha}}$ such that, for any constant $\eps\in(0,1]$,
    \begin{enumerate}[(i)]
      \item\label{theo:hierarchy:theorem:ub} there exists a $(k,\tildeO{k})$-round-adaptive (one-sided) tester for $\property_k$; yet
      \item\label{theo:hierarchy:theorem:lb} any $(k-1,q)$-round-adaptive (two-sided) tester for $\property_k$ must satisfy $q=\bigOmega{n}$.
    \end{enumerate}
\end{theorem}
\noindent We remark that, in fact, the algorithm shown in the first item of \autoref{theo:hierarchy:theorem} also gives an upper bound for the more restricted model of \emph{tail adaptivity}. Specifically, for every $k$ there also exists an $(O(k),\tildeO{k})$-tail-adaptive (one-sided) tester for $\property_k$. Since a $(k-1,q)$-round-adaptive lower bound implies a $(k-1,q)$-tail-adaptive lower bound (see discussion in \autoref{sec:def}), this implies an adaptivity hierarchy (albeit slightly weaker than in \autoref{theo:hierarchy:theorem}) with respect to tail-adaptive testers.

Hereafter we assume, without loss of generality,\footnote{If $n$ is not prime, we choose a prime $p$ such that $n \le p \le 2p$, and use standard padding techniques.}  that $n$ is a prime number, and consider $\field{n}$, the field of order $n$. We will consider the following sequence of ``$k$-iterated address'' functions $(f_k)_{k\geq 0}$ from $\field{n}^n$ to $\bool$, which will in turn lead to the definition of the properties $(\property_k)_{k\geq 0}$ that we use to show the hierarchy theorem. Loosely speaking, $f_k$ receives a vector $x$ of $n$ pointers (indices in $[n]$) and indicates whether when jumping from pointer to pointer $k$ times, starting from an arbitrarily predetermined pointer, we reach a location in which $x$ takes an even value.

To formally define the foregoing functions, first consider $g\colon \field{n}^n\times \field{n} \to \field{n}$ given by $g(x,a) = x_{a+1}$; that is, $g$ returns the coordinate of $x\in\field{n}^n$ ``pointed to'' by $a\in\{0,\dots,n-1\}$. Based on this, we define the iterated versions of $g$, $g_0,\dots,g_n, \dots\colon \field{n}^n \to \field{n}$, as
\begin{align*}
  g_0(x) &= g(x,0)\\
  g_k(x) &= g(x, g_{k-1}(x)) \tag{$k\geq 1$}.
\end{align*}
Finally, we define the \emph{$k$-iterated address} function $f_k\colon\field{n}^n \to \field{n}$ by
\begin{align*}
  f_k(x) = \indic{g_k(x)\text{ even}} = \begin{cases}
  1 & \text{ if } g_k(x) \text{ even}\\
  0 & \text{  otherwise.}
  \end{cases}
\end{align*}
(For instance, $f_0(x)=1$ if and only if $x_1$ is even; and $f_1(x)=1$ if and only if the coordinate of $x$ pointed to by $x_1$, that is $x_{x_1+1}$, is even.) We proceed to describe the outline of the proof of~\autoref{theo:hierarchy:theorem}.

\subsection{High-Level Overview}
Broadly speaking, our roadmap for proving~\autoref{theo:hierarchy:theorem} consists of two main steps:
\begin{enumerate}
  \item We first consider the adaptivity hierarchy question in the setting of randomized \emph{decision tree} (DT) complexity (see~\autoref{sec:DT}). We can view a randomized DT for computing a function $f$ as a probabilisitic algorithm that is given query access to an input $x$ and is required to output $f(x)$ with high probability.  Adapting the definition of round adaptivity (\autoref{def:ra:pt}) in the natural way to decision trees, we will prove the randomized DT analogue of our adaptivity hierarchy theorem, using the foregoing family of address functions $(f_k)_{k\geq 0}$. Namely, we prove that for any $k \geq 0$ with $k=o(n)$, it holds that (i) $f_k$ can be computed by an algorithm making $k+1$ queries, in $k$ adaptive rounds; but (ii) any algorithm using only $k-1$ rounds of adaptivity  must make $\bigOmega{n}$ queries.
  \item We then show a bidirectional connection between \emph{adaptivity-bounded} randomized DT and property testers, which extends the connection observed by Tell~\cite{Tell:14}. This allows us to ``lift'' the DT adaptivity hierarchy theorem to property testing. Specifically, we provide two blackbox reductions between the DT problem of computing function $f$ and property testing for a related property $\property_{f}$, which preserve both the number of adaptive rounds and (roughly) the number of queries. We remark these reductions strongly rely on high-rate codes that exhibit both strong local testability and relaxed local decodability.
\end{enumerate}

The caveat with the above is that to ``lift'' DT lower bounds to testing algorithms via our methodology, we actually need to show lower bounds on a stronger model of DT (this stems from the reductions of the second item, in which we will encode the input via linear codes, requiring the DT algorithm to compute coordinates of this encoding).

Hence, we will actually work in the \emph{linear decision tree} (LDT) model, wherein the algorithm is allowed to query any linear combination (over $\field{n}$) of the coordinates, instead of only querying individual coordinates. (We note that in the case of $\field{2}$, this corresponds to the \emph{parity decision tree} model.) That is, we will proceed as follows:
\begin{enumerate}
  \item \textsf{(L)DT hierarchy:} show that for any $k\geq 0$, the function $f_k$ (i) can be computed by an efficient $(k,O(k))$-round-adaptive (deterministic) DT algorithm, but (ii) does not admit any $(k-1,o(n))$-round-adaptive (randomized, two-sided) LDT algorithm;
  \item \textsf{Transference lemmas:} Show that for any function $f\colon\field{n}^n \to \field{n}$, there exists a property $\mathcal{C}_f \subseteq \field{n}^{m(n)}$ such that, for any $k\geq 0$,
      \begin{enumerate}
          \item a $(k,q)$-round-adaptive testing algorithm for $\mathcal{C}_f$ implies a $(k,q)$-round-adaptive LDT algorithm for $f$ (\autoref{lemma:ldt:pt}).
          \item a $(k,q)$-round-adaptive DT algorithm for $f$ implies a $(k,\tildeO{q})$-round-adaptive testing algorithm for $\mathcal{C}_f$ (\autoref{lemma:pt:dt}).
      \end{enumerate}
\end{enumerate}
Combining the items above will directly imply our hierarchy theorem for property testing (\autoref{theo:hierarchy:theorem}):

\begin{proofof}{\autoref{theo:hierarchy:theorem}}
The upper bound~\ref{theo:hierarchy:theorem:ub} follows immediately from~\autoref{claim:dt:ub} and~\autoref{lemma:pt:dt}, while combining~\autoref{lemma:dt:lb} and~\autoref{lemma:ldt:pt} establishes the lower bound~\ref{theo:hierarchy:theorem:lb}.
\end{proofof}

\paragraph{Organization for the rest of the section.} In~\autoref{sec:DT}, we define the decision tree models and complexities that we shall need. Then, in~\autoref{sec:DThierarchy}, we prove the adaptivity hierarchy theorem for randomized (linear) decision trees. Finally, in~\autoref{sec:thereandbackagain} we prove the transference lemmas that allow us to lift the foregoing hierarchy theorem to the property testing framework.


\subsection{Decision Tree Zoo}
\label{sec:DT}
We shall need to extend the definitions of several different types of decision tree algorithms (see \cite{BW02} for an extensive survey of decision tree complexity) to the setting of bounded adaptivity.

Recall that a  \textsf{deterministic decision tree} is a model of computation for computing a function $f\colon\domain^n \to \domain$. The decision tree is a rooted ordered $|\domain|$-ary tree. Each internal vertex of the tree is labeled with a value $i \in \{1,\dots,n\}$ and the leaves of the tree are labeled with the elements in $\domain$.  Given an input $x \in \domain^n$, the decision tree is recursively evaluated by choosing to recurse on the $i$'th subtree in the $j$'th level if and only if $x_j=i$. Once a leaf is reached, we output the label of that leaf and halt.

Equivalently, we can view deterministic decision trees as algorithms that get oracle access to an input $x \in \domain^n$, then adaptively make queries to $x$, to the end of computing $f(x)$. (Note that the $j$'th query corresponds to the $j$'th layer of the corresponding decision tree, and that the different vertices in the $j$'th layer represent the choices of the next queries, with respect to the answers obtained for previous queries). We define the \emph{deterministic decision tree complexity} of a function $f$ to be the minimal number of queries a deterministic decision tree algorithm needs to make to compute $f$ in the worst case.\footnote{We remark that this definition corresponds to the depth the of decision tree, and not to the number of vertices or edges in the tree.}

Taking the algorithmic perspective, we define $k$-round-adaptive deterministic decision tree algorithms as algorithms that generate their queries in $k$ rounds, where queries in each round may depend on queries from previous rounds. The extension of the foregoing definition to \emph{randomized} decision tree algorithms is done in the natural way, by allowing the algorithm to toss random coins and succeed with high probability (say, $2/3$) in computing $f(x)$. Finally, we shall also extend the definition to linear decision trees, which are decision trees algorithms wherein each query is a linear combination of the elements of the domain. We remark that linear decision trees can be thought of as generalizing both \emph{parity decision trees} and \emph{algebraic query complexity algorithms}~\cite{AW:08}.

More accurately, the aforementioned notions are defined below. We provide the definition of the most general model and derive the more restricted models as special cases.

\begin{definition}[Round-Adaptive Decision Tree Algorithms]\label{def:ra:dt}
Let $\F$ be a finite field of cardinality $n$, and let $k,q \le n$. A (randomized) algorithm $D$ is said to be a \emph{$(k,q)$-round-adaptive} (linear) decision tree algorithm for computing a function $f\colon\F^n \to \F$ if, granted query access to a string $x\in\F^n$, the following holds. 
  \begin{enumerate}[(i)]
        \item \textsf{Query Generation}: The algorithm proceeds in $k+1$ rounds, such that at round $\ell\geq 0$, it produces a set of (linear) queries $Q_\ell\eqdef\{L_{\ell,1},\dots,L_{\ell,\abs{Q_\ell}}\}$, where $L_{\ell,j} \in \F^n$ specifies a linear combination, based on its internal randomness and the answers to the previous sets of queries $Q_0,\dots, Q_{\ell-1}$, and receives the answers $ \langle L_{\ell,1} , x \rangle ,\dots, \langle L_{\ell,\abs{Q_\ell}} , x \rangle$.
    \item \textsf{Computation}: The algorithm computes $f(x)$ with high probability using the answers it received in all $k$ rounds; that is, $\Pr[D^x = f(x)] \ge 2/3$.
  \end{enumerate}
  The \emph{query complexity} $q$ of the tester is the total number of (linear) queries made to $f$, i.e., $q = \sum_{\ell=0}^{k} \abs{Q_\ell}$. The \emph{randomized $(k,q)$-round-adaptive linear decision tree complexity} of a function $f$, denoted $\randLDT{k}{f}$, is the minimal query complexity for a  $(k,q)$-round-adaptive randomized linear decision tree algorithm that computes $f$.
  
  If for all $\ell \in [k+1]$ and $j\in[\abs{Q_\ell}]$ the linear combination $L_{\ell,j}$ only includes a single element (i.e., $L_{\ell,j}$ only has a single non-zero entry), we say that $D$ is a \emph{randomized $(k,q)$-round-adaptive decision tree algorithm complexity}, and denote its corresponding complexity by $\randDT{k}{f}$. If, in addition, the algorithm does not toss any random coins and succeeds with probability $1$, we say that $D$ is a \emph{deterministic $(k,q)$-round-adaptive decision tree algorithm complexity}, and denote its corresponding complexity by $\detDT{k}{f}$.
\end{definition}

\subsection{Decision Tree Hierarchy: Some Things Only Adaptivity Can Address}
\label{sec:DThierarchy}
We first establish the upper bound part of our adaptivity hierarchy theorem for DT, which follows immediately from the construction.
\begin{claim}\label{claim:dt:ub}
  For every $k\geq 0$, there exists a $(k,k+1)$-round-adaptive (deterministic) DT algorithm which computes $f_k$; that is, $\detDT{k}{f_k} \leq k+1$.
\end{claim}
\begin{proof}
The algorithm is straightforward: on input $x\in\field{n}^n$, it sequentially queries $x_1=g_0(x)$, $x_{g_0(x)+1}=g_1(x)$, \dots, $x_{g_{k-1}(x)+1}=g_k(x)$; and returns $1$ if $g_k(x)$ is even, and $0$ otherwise. By definition of $f_k$, this always correctly computes the function, is deterministic, and clearly satisfies the definition of a $(k,k+1)$-round-adaptive DT algorithm.
\end{proof}

We proceed to show the lower bound part of our adaptivity hierarchy theorem for DT, which is proven via a reduction from communication complexity.

\begin{lemma}\label{lemma:dt:lb}
  There exists an absolute constant $c>0$ such that the following holds. For every $0\leq k\leq c\left(\frac{n}{\log n}\right)^{1/3}$, there is no $(k,o(n/(k^2\log n)))$-round-adaptive (randomized) LDT algorithm which computes $f_{k+1}$; that is, $\randLDT{k}{f_{k+1}}=\bigOmega{n/(k^2\log n)}$.
\end{lemma}
\begin{proof}
We will reduce the computation of $f_{k+1}$ (in $k$ rounds of adaptivity) to a related $k$-round two-party randomized communication complexity problem, the ``pointer-following'' problem introduced by Papadimitriou and Sipser~\cite{PS:82}, and conclude by invoking the lower bound of Nisan and Wigderson~\cite{NW:93} on this problem.

This communication complexity problem between two computationally unbounded players, Alice and Bob, is defined as follows. Let $V_A$ and $V_B$ be two disjoint sets of cardinality $n/2$, and let $v_0\in V_A$ be a fixed element known to both players. The input is a pair of functions $(\chi_A,\chi_B)$, where $\chi_A\colon V_A\to V_B$ and $\chi_B\colon V_B\to V_A$. Alice and Bob are given $\chi_A$ and $\chi_B$ respectively, as well as a common random string, and their goal is to compute $\pi_k(\chi_A,\chi_B)\eqdef \chi^{(k)}(v_0)$ with high probability, where $\chi^{(\ell)}$ is the $\ell$-iterate of the function $\chi$:
\begin{align*}
  \chi\colon V_A\cup V_B &\to V_A\cup V_B\\
  v &\mapsto 
  \begin{cases}
    \chi_A(v) & v\in V_A\\
    \chi_B(v) & v\in V_B.
  \end{cases}
\end{align*}
(In other terms, one can see the communication problem as Alice and Bob sharing the edges of a bipartite directed graph where each node has out-degree exactly one, and the goal is to find at which vertex the path of length $k$ starting at a prespecified vertex $v_0$, on Alice's side, ends.)

We will rely on the following lower bound on the $k$-round, randomized (public-coin) version of this problem.
\begin{theorem}[{\cite{NW:93}, rephrased}]\label{theo:nw:pointer:cc:lb}
Any $k$-round randomized communication protocol for the ``pointer-following'' problem, in which Bob sends the first message, must have total communication complexity $\bigOmega{\frac{n}{k^2} - k\log n}$, even to only compute a single bit of $\pi_k(\chi_A,\chi_B)$ with probability at least $2/3$.
\end{theorem}
Note that as long as  $k \ll \left(\frac{n}{\log n}\right)^{1/3}$, this lower bound is $\bigOmega{\frac{n}{k^2}}$. We remark that the fact that the lower bound still holds even when only a single bit of the answer is to be computed will be crucial for us, as our goal is to reduce the communication complexity problem of ``pointer-following'' to computing the Boolean function $f_{k+1}$ in the randomized decision tree model.

Let $\Algo$ be any $(k,q)$-round-adaptive (randomized) LDT algorithm computing $f_{k+1}$. Writing $V_A=\{v_0,\dots, v_{\frac{n}{2} - 1}\}$ and $V_B=\{u_0,\dots, u_{\frac{n}{2}-1}\}$, fix a bijection between $V\eqdef V_A\cup V_B$ (of size $n$) and $\field{n}$ mapping $v_0$ to $1$, so that we identify $V$ with $\field{n}$. On input $(\chi_A,\chi_B)$, Alice and Bob implicitly define the element $x\in\field{n}^{n}$
 by
 $x_1 = \chi_A(v_0)$, $x_2 = \chi_A(v_1)$, \dots, $x_\frac{n}{2} = \chi_A(v_{\frac{n}{2}-1})$ and $x_{\frac{n}{2}+1} = \chi_B(u_0)$, $x_{\frac{n}{2}+2} = \chi_A(u_1)$, \dots, $x_{n} = \chi_A(u_{\frac{n}{2}-1})$. From this, we get that $\pi_{k+2}(\chi_A,\chi_B)=g_{k+1}(x)$, recalling that $g_k(x) = g(x, g_{k-1}(x))$ is recursively defined for $k\geq 1$, and \new{$g_0(x) = x_1$}. Hence deciding whether
  $\pi_{k+2}(\chi_A,\chi_B)$ is even is exactly equivalent to computing $f_{k+1}(x)$.
 
 Alice and Bob can then simulate the execution of $\Algo$ as follows. Without loss of generality, assume it is Alice's turn to speak. To answer a query of the form $\phi_S(x)=\sum_{i\in S} x_i$, she computes $\sum_{i\in S\cap V_A} x_i$ and sends it to Bob; on his side, Bob computes $\sum_{i\in S\cap V_B} x_i$, and receiving Alice's message can then recover the value $\phi_S(x)$ and feed it to the algorithm. (In the next round, when sending his side of the (new) queries to Alice, Bob will also send this value $\phi_S(x)$, to make sure that both sides know the answers to all queries so far.) Since all queries of a given adaptive round of $\Algo$ can be prepared and sent in parallel (costing $O(\log n)$ bits of communication per query), this simulation can be performed in $k+1$ rounds (as many as $\Algo$ takes) with communication complexity $\bigO{q\log }$. At the end, whichever of Alice and Bob received the latest message holds the answer (to ``is $\pi_{k+1}(\chi_A,\chi_B)$ an even node?''), which by assumption on $\Algo$ is correct with probability at least $2/3$. Alice and Bob then use an extra round of communication to broadcast the answer to the other party, bringing the total number of rounds to $k+2$.

 But by~\autoref{theo:nw:pointer:cc:lb}, computing this bit of $\pi_{k+2}(\chi_A,\chi_B)$ with only $k+2$ rounds of communication (Bob speaking first) requires $\bigOmega{\frac{n}{k^2}}$ bits of communication, and so we must have $q=\bigOmega{\frac{n}{k^2\log n}}$.
\end{proof}

\subsection{Adaptivity Bounded Testers and Decision Trees: There and Back Again}
\label{sec:thereandbackagain}
In this section we show how to reduce problems in the adaptivity bounded property testing model to problems in the adaptivity bounded (linear) decision tree model, and vice versa. We begin in \autoref{sec:local_code_prelim}, by presenting the required preliminaries regarding error-correction codes. Then, in \autoref{sec:trans}, we prove the ``transference lemmas'' between these models.

\subsubsection{Preliminaries: Locally Testable and Decodable Codes}
\label{sec:local_code_prelim}
Let $k,n\in\N$. A \textsf{code} over alphabet $\Sigma$ with distance $d$ is a function $C\colon\Sigma^k \to \Sigma^n$ that maps \textsf{messages} to \textsf{codewords} such that the distance between any two codewords is at least $d = d(n)$. If $d = \Omega(n)$, $C$ is said to have \textsf{linear distance}. If $\Sigma=\bitset$, we say that $C$ is a \textsf{binary} code. If $C$ is a linear map, we say that it is a \textsf{linear} code. The \textsf{relative distance} of $C$, denoted by $\delta(C)$, is $d/n$, and its \textsf{rate} is $k/n$. When it is clear from the context, we shall sometime abuse notation and refer to the code $C$ as the set of all codewords $\{C(x)\}_{x \in \Sigma^k}$. Following the discussion in the introduction, we define locally testable codes and locally decodable codes as follows.
\begin{definition}[Locally Testable Codes]
  A code $C\colon\Sigma^k \to \Sigma^n$ is a \textsf{locally testable code} ($\LTC$) if there exists a probabilistic algorithm (tester) $T$ that makes $O(1)$ queries to a purported codeword $w \in \Sigma^n$ and satisfies:
  \begin{enumerate}
\item \textsf{Completeness}: For any codeword $w$ of $C$ it holds
    that $\Pr_T[T^w = 1] \ge 2/3$.
  \item \textsf{Strong Soundness}: For all $w \in \Sigma^{n}$,
    \begin{equation*}
      \Pr_T[T^w = 0] \geq \poly\big(\dist{w}{C}\big).
    \end{equation*}
  \end{enumerate}
\end{definition}

\begin{definition}[Locally Decodable Codes]\label{def:ldc}
  A code $C\colon\Sigma^k \to \Sigma^n$ is a \textsf{locally decodable code} ($\LDC$) if there exists a constant $\decrad \in (0,\delta(C)/2)$ and a probabilistic algorithm (decoder) $D$ that, given oracle access to $w \in \Sigma^{n}$ and direct access to index $i\in[k]$, satisfies the following condition: For any $i\in [k]$ and $w \in \Sigma^{n}$ that is $\decrad$-close to a codeword $C(x)$ it holds that $\Pr[D^{w}(i) = x_i] \geq 2/3$. The \textsf{query complexity} of a $\LDC$ is the number of queries made by its decoder.
\end{definition}

We shall also need the notion of $\rLDC$s (introduced in~\cite{BGHSV:06}). Similarly to $\LDC$s, these codes have decoders that make few queries to an input in attempt to decode a given location in the message. However, unlike $\LDC$s, the relaxed decoders are allowed to output a special symbol that indicates that the decoder detected a corruption in the codeword and is unable to decode this location. Note that the decoder must still avoid errors (with high probability).\footnote{The full definition of $\rLDC$s, as defined in~\cite{BGHSV:06} includes an additional condition on the success rate of the decoder. Namely, for every $w\in\bitset^n$ that is $\decrad$-close to a codeword $C(x)$, and for at least a $\rho$ fraction of the indices $i\in[k]$, with probability at least $2/3$ the decoder $D$ outputs the $i$'th bit of $x$. That is, there exists a set $I_w\subseteq [k]$ of size at least $\rho k$ such that for every $i\in I_w$ it holds that $\Pr\left[D^w(i)=x_i\right]\geq 2/3$.
	We omit this condition since it is irrelevant to our application, and remark that every $\rLDC$ that satisfies the first two conditions can also be modified to satisfy the third conditions (see~\cite[Lemmas 4.9 and 4.10]{BGHSV:06}).}

\begin{definition}[Relaxed-LDC]\label{def:rldc}
  A code $C\colon\Sigma^k \to \Sigma^n$ is a $\rLDC$ if there exists a constant $\decrad \in (0,\delta(C)/2)$ such that the following holds.
  \begin{enumerate}
  \item \textsf{(Perfect) Completeness}: For any $i\in [k]$ and $x\in\Sigma^k$  it holds that $D^{C(x)}(i) = x_i$.
  \item \textsf{Relaxed Soundness}: For any $i\in [k]$ and any $w \in\Sigma^{n}$ that is $\decrad$-close to a (unique) codeword $C(x)$, it holds that
    \begin{equation*}
      \Pr[D^{w}(i) \in \{x_i,\bot\}] \geq 2/3. 
    \end{equation*}
  \end{enumerate}
\end{definition}
\noindent There are a couple of efficient constructions of codes that are both $\rLDC$s and $\LTC$s (see~\cite{BGHSV:06,GGK:15}). We shall need the construction in~\cite{GGK:15}, which has the best parameters for our setting.\footnote{Specifically, the codes in~\cite{GGK:15} are meaningful for every value of the proximity parameter, whereas the codes in~\cite{BGHSV:06} require $\eps > 1/\poly\!\log(k)$.}

\begin{theorem}[e.g., {\cite[Theorem 1.1]{GGK:15}}]\label{theo:ssec:nice:codes}
	\label{thm:local_codes}
	For every $k\in\N$, $\alpha>0$, and finite field $\field{}$ there exists an $\field{}$-linear code $C\colon\field{}^k \to \field{}^{k^{1+\alpha}}$ with linear distance, which is both a $\rLDC$ and a (one-sided error) $\LTC$ with query complexity $\poly(1/\eps)$; furthermore, both testing and (relaxed) decoding procedures are non-adaptive.
\end{theorem}

\subsubsection{Transference Lemmas}
\label{sec:trans}
Fix any $\alpha > 0$. Let $C\colon\field{n}^n\to\field{n}^m$ be a code with constant relative distance $\delta(C) > 0$, with the following properties:
\begin{itemize}
  \item \emph{linearity}: for all $i\in[m]$, there exists a set $S_i \subseteq [n]$ such that $C(x)_i = \sum_{j\in S_i} x_i$ for all $x\in\field{n}^n$;
  \item \emph{rate}: $m\leq n^{1+\alpha}$;
  \item \emph{testability}: $C$ is a \emph{strong-LTC} with one-sided error and \emph{non-adaptive} tester;
  \item \emph{decodability}: $C$ is a \emph{relaxed-LDC}.\end{itemize}
We will rely on~\autoref{theo:ssec:nice:codes} for the existence of such codes. Before delving into the details, we briefly explain the reason for each of the points above. The linearity will be crucial to reduce to and from the LDT model: indeed, any coordinate of a codeword corresponds to a fixed linear combination of the coordinates of the message, which corresponds to a single LDT query on that particular linear combination. The rate bound is required since our lower bounds are in terms of the dimension $n$ and upper bounds in terms of the block-lengh $m$. Ideally, we would like $m=O(n)$, to have a direct correspondence between the LDT and the property testing query complexities; however, this nearly-linear rate is the best known achievable for constant-query LTCs and relaxed-LDCs~\cite{GGK:15}. The LTC property will be useful to us in the reduction from property testing to DT query complexity (where we will need to first check that our input is close to a codeword, in view of decoding the closest message during the reduction), where the \emph{strong} testability (i.e., rejection with probability proportional to the distance from a valid codeword) will allow us do deal with arbitrarily small values of the proximity parameter. Similarly, we will rely on the (relaxed) LDC property in that same reduction, in order to obtain individual coordinates of the message, given query access to an input close to a codeword.

\noindent We proceed to show the framework for reducing property testing to decision tree complexity and vice-versa. For a fixed function $f\colon\field{n}^n\to\bool$, consider the subset $f^{-1}(1) \subseteq \field{n}^n$; and define the sets of codewords $\mathcal{C}\eqdef C(\field{n}^n)\subseteq \field{n}^m$,  $\mathcal{C}_f\eqdef C(f^{-1}(1)) = \setOfSuchThat{ C(x) }{ x\in \field{n}^n,\ f(x)=1 }\subseteq \mathcal{C}$.

Consider now testing the property $\mathcal{C}_f$: we will reduce the LDT computation of $f$ to the testing of $\mathcal{C}_f$. Specifically, we prove the following.

\begin{lemma}[LDT $\leadsto$ PT Reduction Lemma]\label{lemma:ldt:pt}
Fix any $f\colon\field{n}^n\to\bool$. If there exists an $(k,q)$-round-adaptive tester for $\mathcal{C}_f$, then there is an $(k,q)$-round-adaptive LDT algorithm for $f$.
\end{lemma}
\begin{proof}
Suppose there exists a $(k,q)$-round-adaptive tester $\Tester$ for $\mathcal{C}_f$. On input $x\in\field{n}^n$, we emulate the invocation of $\Tester$, with respect to proximity parameter $\eps = \delta(C)$, on the encoded input $y\eqdef C(x) \in \field{n}^m$ and output $1$ if and only if $\Tester$ returns \accept. To see why this is correct, observe that by definition, if $f(x)=1$ then $y\in\mathcal{C}_f$. However, if $f(x)=0$, then for any $y'\in\mathcal{C}_f$  such that $y' = C(x)$ we must have $\dist{y}{y'}>\eps$, by the distance of our code.

It remains to show that this simulation can be achieved efficiently, as claimed. To do so, we will rely on the fact that $C$ is a linear code: whenever $\Tester$ queries $y_i$, we can compute the set $S_i\subseteq [n]$ (which only depends on $C$, and not on $x$), and perform the LDT query $\sum_{j\in S_i} x_j$. The simulation clearly preserves the number of adaptive rounds as well, concluding the proof.
\end{proof}

In our next lemma, we give a partial converse relating property testing and decision tree complexity, with some logarithmic overhead in the resulting query complexity.
\begin{lemma}[PT $\leadsto$ DT Reduction Lemma]\label{lemma:pt:dt}
Fix any $f\colon\field{n}^n\to\bool$. If there exists an $(k,q)$-round-adaptive (randomized) DT algorithm for $f$, then there is a $(k,O(q \log q) + \poly(1/\eps))$-round-adaptive tester for $\mathcal{C}_{f}$. (Moreover, if the DT algorithm is always correct, then this tester is one-sided.)
\end{lemma}
\begin{proof}
Fix $k\geq 0$, and suppose there exists such a $(k,q)$-round-adaptive DT algorithm $\Algo$ for $f$. On input $y\in\field{n}^m$ and proximity parameter $\eps\in(0,1]$, we would like to decode $y$ to a message $x\in\field{n}^n$ and invoke the algorithm on $x$ to determine if $f(x)=1$; more precisely, we wish to invoke the DT algorithm while simulating each query to $x$ by locally decoding $y$ using  $O(1)$ queries. The issue, however, is that the success of the local decodable is only guaranteed for inputs that are sufficiently close to a valid codeword, and we have no such guarantee on $y$ \textit{a priori}. However, recalling that $C$ is a strong-LTC, we can handle this as follows. Letting $\decrad>0$ be the decodability radius of the relaxed-LDC $C$, we set $\delta^\ast \eqdef \min(\decrad,\eps)$.
\begin{enumerate}[(1)]
  \item\label{ptldt:step:1} Run independently $O(\poly(1/\delta^\ast))$ times the local tester for the strong-LTC $C$ on $y$, and output \reject if any of these rejected. Since every invocation of the local tester makes $O(1)$ queries to $y$, this has query complexity $O(\poly(1/\delta^\ast)) = O(\poly(1/\eps))$; and if $\dist{y}{\mathcal{C}} > \delta^\ast$ then this step outputs $\reject$ with probability at least $9/10$.
  \item\label{ptldt:step:2} Invoke $\Algo$ on the message $x\eqdef\arg\!\min \setOfSuchThat{\dist{C(x)}{y}}{ x\in\field{n}^n }$, answering each query $x_i$ by calling the local decoder for the relaxed-LDC $C$. This is done so that the decoder is correct with probability at least $1/(10q)$, by standard repetition (taking the plurality value); with the subtlety that we output $\reject$ immediately whenever the decoder returns $\bot$. Since each query can be simulated by $O(\log(q))$ queries (repeating the $O(1)$ queries of the decoder $O(\log(1/q))$ times), this step has query complexity $O(q\log q)$; and at the end, we output $\accept$ if, and only if, $\Algo$ returns the value $1$ for $f(x)$.
\end{enumerate}
Importantly, Step~\ref{ptldt:step:1} can be run in parallel to Step~\ref{ptldt:step:2}, and in particular can be executed during the first ``batch'' of queries $\Algo$ makes. This guarantees that the whole simulation above uses the same number of adaptive rounds as $\Algo$, as claimed. It remains to argue correctness.

\paragraph{Completeness.} Assume $y\in \mathcal{C}_{f}$. In particular, $y$ is a codeword of $C$, and the (one-sided) local tester returns \accept with probability one in~\ref{ptldt:step:1}. Then, since by definition there is a unique $x\in \field{n}^n$ such that $C(x)=y$, the local decoder of Step~\ref{ptldt:step:2} will correctly output the correct answer for each query with probability $1$, and therefore $\Algo$ will correctly output $f(x)$ with probability $2/3$ -- so that the tester returns $\accept$ with probability at least $2/3$ overall. (Moreover, if the DT algorithm $\Algo$ always correctly compute $f$, then the tester returns \accept with probability one.)

\paragraph{Soundness.} Assume $\dist{y}{\mathcal{C}_{f}} > \eps$. If $\dist{y}{\mathcal{C}} > \delta^\ast$, then the local tester returns \reject with probability at least $9/10$ in Step~\ref{ptldt:step:1}. Therefore, we can continue assuming that $\dist{y}{\mathcal{C}} \leq \delta^\ast$, which satisfies the precondition of the relaxed-LDC decoder in Step~\ref{ptldt:step:2}. By a union bound over all $q$ queries, with probability at least $9/10$ we have that the decodings performed in Step~\ref{ptldt:step:2} are all correct; in which case we answer the queries of the algorithm according to $x\eqdef\arg\!\min \setOfSuchThat{\dist{C(x)}{y}}{ x\in\field{n}^n }$ (or possibly answered by $\bot$, in which case the tester immediately outputs \reject and we are done). Since $\dist{y}{C(x)} \leq \delta^\ast \leq \eps$, we must have  $C(x)\not\in\mathcal{C}_f$, which implies that $\Algo$ correctly returns $f(x)=0$ with probability at least $2/3$, in which case the tester outputs $\reject$. Overall, this happens with probability at least $9/10\cdot 9/10\cdot 2/3 = 27/50$.\medskip

Thus, in both cases the tester is correct with probability at least $27/50$; repeating a constant number of times (as explained in~\autoref{rk:repetition}) and taking the majority vote allows us to amplify the probability of success to $2/3$.
\end{proof}

\section{An Adaptivity Hierarchy with respect to a Natural Property}\label{sec:hierarchy:natural}
\newcommand{\disjcycle}[2]{\Sigma_{#1,#2}}
\newcommand{\indepset}[1]{\textsc{Is}_{#1}}

In this section we show a \emph{natural} property of graphs for which, broadly speaking, more adaptivity implies more power. More specifically, we prove the following adaptivity hierarchy theorem with respect to the property of $k$-cycle freeness in the bounded-degree graph model (see definitions in~\autoref{sec:cycle_prelim}).

\begin{theorem}
\label{thm:natural_hierarchy}
	Let $k\in \N$ be a constant. Then,
	 \begin{enumerate}[(i)]
      \item there exists a $(k,O(1/\eps))$-round-adaptive (one-sided) tester for $(2k+1)$-cycle freeness in the bounded-degree graph model; yet
      \item any $(k-1,q)$-round-adaptive (two-sided) tester for $(2k+1)$-cycle freeness in the bounded-degree graph model must satisfy $q=\bigOmega{\sqrt{n}}$.
      \end{enumerate}
\end{theorem}

We stress that although~\autoref{theo:hierarchy:theorem} establishes an adaptivity hierarchy with stronger separations, the merit of \autoref{thm:natural_hierarchy} is in showing that an adaptivity hierarchy also holds for a \emph{natural} well-studied property. We further observe that the choice of the bounded-degree graph model is not insignificant: one cannot hope to establish such a striking gap in other settings such as the dense graph model or in the Boolean function testing setting. Indeed, as discussed in~\autoref{ssec:previous:work} it is well-known that in these two models, any adaptive tester can be made (fully) non-adaptive at the price of only a quadratic and exponential blowup in the query complexity, respectively(see~\cite{AFKS00,GT03} for the former; the latter is folklore). We remark that in \autoref{ssec:misc:reducing:rounds} we discuss emulating testers with $k$ rounds of adaptivity by testers with $k'<k$ rounds.

\subsection{Cycle Freeness in the Bounded Degree Graph Model}
\label{sec:cycle_prelim}

In the subsection we provide the necessary definitions and establish a basic upper bound on the complexity of $k$-adaptive testing of cycle freeness in the bounded degree graph model. We begin with a definition of the model.

Let $G=(V,E)$ be a graph with constant degree bound $d < |V|$, represented by its adjacency list; that is, represented by a function $g: V \times d \to V$ such that $g(v,i) = u \in V$ if $u$ is the $i$th neighbor of $v$ and $g(v,i) = 0$ if $v$ has less than $i$ neighbors. A bounded degree graph property $\property$ is a subset of graphs (represented by their adjacency list) that is closed under isomorphism; that is, for every permutation $\pi$ it holds that $G \in \property$ if and only if $G \in \pi(G)$. The distance of graph $G$ from property $\property$ is the minimal fraction of entries in $g$ one has to change to reach an element of $\property$.

We extend the definition of functional round-adaptive testing algorithms to the bounded degree graph model in the natural way.

\begin{definition}[Round-Adaptive Testing in the Bounded Degree Graph Model]
Let $G=(V,E)$ be a graph with constant degree bound $d < |V|$, represented by its adjacency list $g\colon V \times d \to V$, and let $k,q \le n$. A randomized algorithm is said to be a \emph{$(k,q)$-round-adaptive} tester for a (bounded degree) graph property $\property$, if, on proximity parameter $\eps\in(0,1]$ and granted query access to $g$, the following holds. 
  \begin{enumerate}[(i)]
        \item \textsf{Query Generation}: The algorithm proceeds in $k+1$ rounds, such that at round $\ell\geq 0$, it produces a set of queries $Q_\ell\eqdef\{x^{(\ell),1},\dots,x^{(\ell),\abs{Q_\ell}}\}\subseteq \domain$ (possibly empty), based on its own internal randomness and the answers to the previous sets of queries $Q_0,\dots, Q_{\ell-1}$, and receives $f(Q_\ell)=\{g(x^{(\ell),1}),\dots,g(x^{(\ell),\abs{Q_\ell}})\}$;
    \item \textsf{Completeness}: If $G\in\property$, then the algorithm outputs $\accept$ with probability at least $2/3$;
    \item \textsf{Soundness}: If $\dist{G}{\property} > \eps$, then the algorithm outputs $\reject$ with probability at least $2/3$.
  \end{enumerate}
  The \emph{query complexity} $q$ of the tester is the total number of queries made to $f$, i.e., $q = \sum_{\ell=0}^{k} \abs{Q_\ell}$. If the algorithm returns \accept with probability one whenever $f\in\property$, it is said to have \emph{one-sided} error (otherwise, it has \emph{two-sided} error). As before, we will sometimes refer to a tester with respect to proximity parameter $\eps$ as an $\eps$-tester.
\end{definition}

Next, we define the (bounded degree) graph property of $k$-cycle freeness.

\begin{definition}[Cycle Freeness]
	Let $k\in\N$. A graph $G=(V,E)$ is said to be $k$\textsf{-cycle free} if it does not contain any cycle of length less or equal to $k$; that is, if for every $t \leq k$ and $v_1,\ldots,v_{t} \in V$ either $(v_t,v_1) \not\in E$ or there exists $i\in[t-1]$ such that $(v_i,v_{i+1}) \not\in E$.
\end{definition}
Finally, we make the following observation, which roughly speaking implies that when surpassing a certain threshold of round adaptivity, testing cycle freeness in the bounded degree graph model becomes ``easy.''\footnote{This is a specific case of a more general algorithm for testing subgraph freeness; see e.g.~\cite[Section 9.2.1]{Gol:17}.}

\begin{observation}
	\label{obser:upper_bound}
  For every $k\in\N$ there exists a $(k,q)$-round-adaptive testing algorithm for $(2k+1)$-cycle freeness and $(2k+2)$-cycle freeness in the bounded-degree graph model with query complexity $q=O(d^{k+1}/\eps)$.
\end{observation}
\begin{proof}
	The algorithm explores the graph in the most natural way: starting from  $O(1/\eps)$ ``source vertices'' selected uniformly at random, it adaptively explore their neighborhoods by querying at each round the neighbors of the previously reached vertices, in a breadth-first-search fashion. If any $(2k+1)$-cycle (resp. $(2k+2)$-cycle) is detected, the algorithm rejects, and accepts otherwise. (Clearly, this tester is one-sided.) It is easy to see that if any of the source vertices belongs to a $(2k+1)$- or $(2k+2)$-cycle, then this bounded-depth BFS will detect it; thus, we only need to argue that if the graph is $\eps$-far from cycle freeness, with constant probability, one of the source vertices will participate in such a cycle. But this is the case, as any such graph must have at least $\eps n$ vertices participating in a cycle (indeed, otherwise one could ``correct'' the graph by removing less than $\eps d n$ vertices, contradicting the distance).
		
	Finally, for each source vertex, after $k$ rounds of adaptivity the number of nodes visited is at most $O(d^{k+1})$, hence the claimed query complexity.
\end{proof}

\subsection{Lower Bounds for Round-Adaptive Testers}
In this subsection, we prove the following lemma, which roughly speaking shows that testing $(2k+3)$-cycle freeness is hard for $k$-round-adaptive testing algorithms. 
\begin{lemma}
\label{lem:cycle_lower_bound}
  Let $k\in\N$ be constant. Then, any $(k,q)$-round-adaptive testing algorithm for $(2k+3)$-cycle freeness in the bounded-degree graph model must satisfy $q=\bigOmega{\sqrt{n}}$.
\end{lemma}

In stark contrast, recall that \autoref{obser:upper_bound} shows that testing $(2k+2)$-cycle freeness is easy for $k$-round-adaptive testing algorithms. Indeed, the proof of \autoref{thm:natural_hierarchy} follows  by combining \autoref{obser:upper_bound} and \autoref{lem:cycle_lower_bound} together.

\begin{proofof}{\autoref{lem:cycle_lower_bound}} We will show a distribution of $(2k+3)$-cycle free graphs, denoted $\dyes$, and a distribution of graphs that are ``far'' from being $(2k+3)$-cycle free, denoted $\dno$, and prove that no $(k,q)$-round-adaptive testing algorithm can distinguish, with high probability, between $\dyes$ and $\dno$. Loosely speaking, $\dyes$ consists of all graphs whose vertices are covered via disjoint $(2k+4)$-cycles, and $\dno$ consists of all graphs whose vertices are covered via disjoint $(2k+3)$-cycles.

More accurately,  denote by $\property_{t,n,d}$ the subset of $n$-node graphs with maximum degree at most $d$ that are $t$-cycle-free. Let $\disjcycle{t}{s}$ be the $2$-regular graph on $st$ vertices made of $s$ disjoint $t$-cycles, namely $(v_1,\dots, v_t)$, $(v_{t+1},\dots, v_{2t})$, $(v_{(s-1)t+1},\dots, v_{st})$. Denote also by $\indepset{r}$ the independent set on $r$ vertices. For two graphs $G,G'$ on respectively $m$ and $m'$ vertices and with $e$ and $e'$ edges, we write $G\disjunion G'$ for the graph on $m+m'$ vertices and with $e+e'$ edges obtained by concatenating disjoint copies of $G,G'$. 

For $k=O(1)$, we let $\ell\eqdef \flr{\frac{n}{(2k+4)}}$, $\ell'\eqdef \flr{\frac{n}{(2k+3)}}$, and define the two distributions over $n$-node graphs $\dyes$ and $\dno$ as follows.
\begin{itemize}
  \item $\dyes$ is the uniform distribution over all isomorphic copies of $G^\yes_k\eqdef \disjcycle{(2k+4)}{\ell}\disjunion\indepset{n-(2k+4)\ell}$;
  \item $\dno$ is the uniform distribution over all isomorphic copies of $G^\no_k\eqdef\disjcycle{(2k+3)}{\ell'}\disjunion\indepset{n-(2k+3)\ell'}$.
\end{itemize}
The next claim establishes that indeed $\dyes$ consists of \yes-instances, whereas $\dno$ consists of \no-instances.

\begin{claim}
  $\dyes$ is supported on $\property_{(2k+3),n,d}$, while every graph in the support of $\dno$ is $\bigOmega{1}$-far from $\property_{(2k+3),n,d}$.
\end{claim}
\begin{proof}
  The first part is obvious, as the only cycles in $G^\yes_k$ are $(2k+4)$-cycles. As for the second, it immediately follows from observing that $G^\no_k$ contains $\ell'$ disjoint $(2k+3)$-cycles, and thus at least $\ell'$ edges have to be removed to make it $(2k+3)$-cycle free. Thus, $\dist{G^\no_k}{\property_{(2k+3),n,d}} \geq \frac{\ell'}{d n/2} = \bigOmega{\frac{1}{dk}} = \Omega_d(1)$. 
\end{proof}

Let \Tester be a deterministic testing algorithm with $k$ rounds of adaptivity and query complexity $q'=\littleO{\sqrt{n}}$. The following lemma concludes the proof of \autoref{lem:cycle_lower_bound} by  showing that  \Tester cannot distinguish, with high probability, between graphs in $\dyes$ and graphs in $\dno$. Denote \Tester's (disjoint) query sets, per round, by $Q_0,\dots, Q_k\subseteq V$, where a query is a vertex $v$. Denote the corresponding sets of answers by $A_0,\dots,A_k$, where the answer to a query $v$ consists of the labels of all neighbors of $v$ (i.e., either two or zero vertices). Since $k=\bigO{1}$, without loss of generality, we can assume (by padding) that all query sets have the same size $q\eqdef \abs{Q_i} = \frac{q'}{k+1} = \bigTheta{q'}$ for every $i\in\{0,\dots,k\}$. Moreover, we can also assume that no vertex is queried twice, i.e. that all $Q_i$'s are disjoint.

\begin{lemma}
	\label{lem:indisting}
  $\abs{ \probaDistrOf{G\sim\dyes}{\Tester^G\text{ accepts}} - \probaDistrOf{G\sim\dno}{\Tester^G\text{ accepts}} } \leq \frac{1}{10}$.
\end{lemma}
\begin{proof}
  For $j\in\{0,\dots,k\}$, define by $Y_j$ and $N_j$ the distribution of $(A_0,\dots,A_j)$ when $G\sim\dyes$ and when $G\sim\dno$, respectively. We shall prove that $\totalvardist{Y_k}{N_k} \leq \frac{1}{10}$, which by the data processing inequality will imply the claim of \autoref{lem:indisting}.
  
  The high-level idea is that in each round, the tester can either query ``fresh'' vertices, of which it has no prior information, or query the boundaries (i.e., the direct neighbors) of previously queried vertices. Then, loosely speaking we can argue that, on the one hand, if the total number of queries is $o(\sqrt{n})$, then both for graphs in $\dyes$ and $\dno$ all queries of ``fresh'' vertices (obtained during all rounds) with high probability would only fall into previously unattained disjoint cycles, in which case the answer would be a uniform sequence of ``fresh'' labels. On the other hand, the local view obtained by querying the boundary, using at most $k$ rounds of adaptive queries, of each vertex previously obtained via a ``fresh'' query (which by the above lies in a cycle wherein the tester has no information of the labels of the other vertices participating in this cycle) is isomorphic to the tail graph over fresh labels, both for instances taken from $\dyes$ and $\dno$ (that is, we do not have enough adaptive queries to observe a full cycle). The foregoing intuition is formalized below.
  
  For $i\in\{0,\dots,k\}$, define
  \begin{align*}
      S_{i}^{\rm f} &\eqdef Q_i \setminus\cup_{j=0}^{i-1} A_j \\
      S_{i}^{\rm b} &\eqdef Q_i \cap\cup_{j=0}^{i-1} A_j 
  \end{align*}
  to be, respectively, the set of ``entirely fresh'' nodes queried at round $i$ (that is, nodes that are not neighbors of any previously queried node), and the set of ``boundary nodes'' (which are the not-yet-queried nodes neighbors of a previously queried node).
  
  First, we bound the probability that any of the $q'$ queries made ``hits'' the set of disconnected nodes:
  \begin{claim}\label{claim:hit:no:isolated}
  Let $E_1(G)$ denote the event that $\Tester$ queries an isolated vertex of $G$, that is $E_1(G) \eqdef \{\exists i,v\text{ s.t. } v\in Q_i,\ \deg(v) = 0\}$. Then 
      $\probaDistrOf{G\sim\dyes}{ E_1(G) }, \probaDistrOf{G\sim\dno}{ E_1(G) } = o(1)$.
  \end{claim}
  \begin{proof}
    This follows by induction: at step $i$, conditioned on no isolated node having been queried yet, the algorithm has degree information about $\abs{\cup_{j=0}^{i-1} Q_j\bigcup\cup_{j=0}^{i-1} A_j} \leq \sum_{j=0}^{i-1} \abs{Q_j}+\sum_{j=0}^{i-1} \abs{A_j} \leq 3q\cdot i$ nodes, so there remain at least $n-3kq$ nodes on which the algorithm has no degree information at all. Among these, there are $n-(2k+4)\ell\leq (2k+4)$ (or $n-(2k+3)\ell'\leq (2k+3)$, in the \no-case) isolated nodes. By symmetry, this means that in the new batch of $q$ queries, the algorithm will query one of these isolated nodes with probability at most $1-\left(1-\frac{(2k+4)}{n-3kq-(2k+4)}\right)^q = 1-\left(1-\frac{O(1)}{n}\right)^{q} = \bigO{\frac{q}{n}} = o(1)$. Therefore, overall there will be an isolated node queried with probability at most $k\cdot o(1) = o(1)$.
  \end{proof}
  
  Next, we argue that at each step, with overwhelming probability all the ``fresh nodes'' queried fall in distinct cycles, which have not been attained yet.
    \begin{claim}\label{claim:hit:no:previous:cycle}
  Let $E_2(G)$ denote the event that at some round $i$, one of the queries in $S_{i}^{\rm f}$ belongs to the same cycle (either a $(2k+4)$- or a $(2k+3)$-cycle, depending on whether the graph is drawn from $\dyes$ or $\dno$) as one of the previous queries $\bigcup_{j=0}^{i-1} Q_j$. Then $\probaDistrOf{G\sim\dyes}{ E_2(G) }, \probaDistrOf{G\sim\dno}{ E_2(G) } = o(1)$.
  \end{claim}
  \begin{proof}
  We will show that $\probaDistrOf{G\sim\dyes}{ E_2(G) }  = o(1)$; the \no-case is similar. For $i\in\{1,\dots,k\}$, let $E^{(i)}_2(G)$ denote the event that at some round $i$, one of the queries in $S_{i}^{\rm f}$ belongs to the same cycle as a previous query, so that $E_2(G) = \bigcup_{i=1}^k E^{(i)}_2(G)$.
  
  Note that since $\lvert \bigcup_{j=0}^{i-1} Q_j\rvert = iq$, we have $\lvert \bigcup_{j=0}^{i-1} A_j\rvert \leq 2iq$ (and the number of distinct cycles reached is at most $\lvert \bigcup_{j=0}^{i-1} Q_j\rvert$).  Therefore, at round $i$ each of the at most $q$ distinct queries in $S_{i}^{\rm f}$ falls independently in a previously visited cycle with probability upper bounded by 
  \[
      \frac{iq\cdot(2k+4)}{n-3iq} \leq \frac{kq\cdot(2k+4)}{n-3kq} \leq \frac{2k^2q}{n}
  \]
  recalling that $q=o(n)$ and $k=O(1)$. A union bound over all at most $q$ queries of $S_{i}^{\rm f}$, and then over the $k$ rounds then shows that $\probaDistrOf{G\sim\dyes}{ E_2(G) } \leq \frac{2k^3q^2}{n} = o(1)$ (since $q = \littleO{\sqrt{n}}$).
  \end{proof}

To conclude the proof, note that by the above, with probability $1-o(1)$ neither $E_1$ nor $E_2$ occurs; that is, none of the isolated vertices was queried, and all the ``fresh'' queries (during all rounds ) fell in previously unattained distinct cycles. In this case, at each round of adaptivity the algorithm can at most discover two new nodes out of every cycle it reached before (by including the one or two end nodes of the current ``discovered portion'' into $S_{i}^{\rm b}$). Therefore, on any cycle ever reached, the  $(k,q)$-round-adaptive testing algorithm can observe at most $2k+2$ nodes (which then form a consecutive path). We show that this implies that the algorithm cannot distinguish between a $\no$-instance and a $\yes$-instance, as loosely speaking, in both cases its local view is of a tail graph over uniformly distributed fresh labels, and so it is unable to determine whether it belongs to a cycle of length $2k+3$ or $2k+4$.

To make the argument more precise, we will actually show a stronger statement; namely, we show that, conditioning on neither $E_1$ nor $E_2$ occuring, a simulator with no access to the graph can answer the queries of the testing algorithm in a way that is indistinguishable from the tuple of answers obtained from querying a graph distributed according to either $\dyes$ or $\dno$. This simulator operates as follows: at round $i$,
\begin{enumerate}
  \item Order (arbitrarily) all the nodes of $Q_i$: $v_1,\dots, v_q$, and initialize the set of available-to-sample nodes $U\gets V\setminus \left( Q_i\cup \bigcup_{j=0}^{i-1} Q_j\cup \bigcup_{j=0}^{i-1} A_j\right)$.
  \item Do sequentially the following, for $s=1\dots q$:
  \begin{itemize}
    \item if $v_s\in S_{i}^{\rm f}$ (fresh node: no previous neighbors known), pick uniformly at random two distinct nodes $u,u'$ in $U_s$ and return them as answers (i.e., declare them as neighbors of $v_s$);
    \item otherwise, $v_s\in S_{i}^{\rm b}$ (boundary node: exactly one already known neighbor, call it $u$): pick uniformly at random one other node $u'$ in $U_s$, and return $(u,u')$ as answers;
    \item update $U$ by removing $u,u'$: $U\gets U\setminus\{u,u'\}$
  \end{itemize}
\end{enumerate}
It is straightforward to verify that, since we conditioned on $\overline{E_1}$ and $\overline{E_2}$, this simulates exactly the same distribution over nodes (over the choice of $G$); since this is the same both for $\dyes$ and $\dno$, we get that
$\totalvardist{(Y_k\mid \overline{E_1\cup E_2}))}{(N_k\mid \overline{E_1\cup E_2}))} = 0$, which combined with~\autoref{claim:hit:no:isolated} and~\autoref{claim:hit:no:previous:cycle} finishes the proof.
\end{proof}
This concludes the proof of \autoref{lem:cycle_lower_bound}.
\end{proofof}

\section{Some Miscellaneous Remarks}
	\label{sec:misc}
\subsection{On Simulating $k$ Rounds With Fewer}\label{ssec:misc:reducing:rounds}

As mentioned in the beginning of ~\autoref{sec:hierarchy:natural}, in the Boolean setting any adaptive property testing algorithm can be simulated non-adaptively with only an exponential blowup in the query complexity. Phrased differently, this implies that any property of Boolean functions which admits a $(k,q)$-round-adaptive tester also has  a $(0,2^q-1)$-round-adaptive tester.

This begs the following more general question: let $\property = \bigcup_{n} \property_n$ be a property of Boolean functions, such that there exists a $(k,q)$-round-adaptive tester for $\property$. For $\ell < k$, what upper bound can we obtain on the query complexity $q'$ of the best $(\ell,q')$-round-adaptive tester for $\property$?

Denoting by $q_\ell$ this query complexity, the above discussion immediately implies:
\begin{fact}\label{fact:round:reduction:alltheway}
  For any $0\leq \ell \leq k$, one has $q_k \leq q_\ell \leq 2^{q_k}-1$.
\end{fact}

In what follows, we provide a example of a more fine-grained version of this fact, in the case when $\ell=k-1$ (that is, one wishes to reduce the number of rounds of adaptivity by one).

\begin{proposition}\label{prop:round:reduction:babystep}
  For any $0 < k$, one has $q_k \leq q_{k-1} \leq q_k(1+2^{\frac{q_k}{k}})$.
\end{proposition}
\begin{proof}
Let $\Tester_k$ be a $(k,q)$-round-adaptive tester for $\property$, which can be viewed as a distribution over deterministic algorithms. Thus, it is sufficient to explain how to simulate any deterministic algorithm with $k$ rounds of adaptivity by one with $\ell$ rounds. Fix such a $(k,q)$-round deterministic algorithm: this can be seen equivalently as a depth-$(k+1)$ binary tree, where each internal node $v$ is labeled by the set of queries $Q_v$ made at that stage, and the leaves are either $\accept$ or $\reject$. By assumption, we have that on each path $(v_0,v_1,\dots,v_k,v^\ast)$ from the root to a leaf, $\sum_{j=0}^k \abs{Q_{v_j}} \leq q$; moreover, one can assume without loss of generality that this is an equality.

The idea is then to contract, on any path, two consecutive nodes as follows: instead of querying $Q_{v_j}$, receiving the answers, and then querying the (adaptively chosen) set $Q_{v_{j+1}}$, one can idea query simultaneously $Q_{v_j}$ and the union of all possible sets $Q_{v_{j+1}}$: since the latter depends only on the previous queries, and the only unknown answers are those to the queries in $Q_{v_j}$, there are at most $2^{\abs{Q_{v_j}}}$ possibilities for $Q_{v_{j+1}}$. As clearly no matter what $Q_{v_{j+1}}$ would be, its size is at most $q$, the set $Q'_i = Q_{v_j} \cup \bigcup_{Q\colon \text{ possible }Q_{v_{j+1}} } Q$ queried has size at most $\abs{Q_{v_j}} + q2^{\abs{Q_{v_j}}}$. Thus, by contradicting the two rounds $i$ and $i+1$, one incurs an additional number of queries upper bounded by $q2^{\abs{Q_{v_j}}} - \abs{Q_{v_{j+1}}} \leq q2^{\abs{Q_{v_j}}}$

By an averaging argument, since on every such path we have $\sum_{j=0}^k \abs{Q_{v_j}} = q$, there must exist an index $j^\ast$ such that $\abs{Q_{v_{j^\ast}}} \leq \frac{q}{k+1}$. Since we would like to ``contract'' rounds $j^\ast$ and $j^\ast+1$ into a single round, we additionally want to ensure $j^\ast < k$. But similarly, as $\sum_{j=0}^{k-1} \abs{Q_{v_j}} \leq q$ there exists $i^\ast$ such that $\abs{Q_{v_{i^\ast}}} \leq \frac{q}{k}$. We then get an index $i^\ast < k$ (which depends on the path taken down the tree) to which we can apply the above transformation. That is, whenever the deterministic algorithm is executed it will reach an index $i^\ast<k$ where it should make $\abs{Q_{v_{i^\ast}}} \leq \frac{q}{k}$ queries. At that point, it makes instead  these queries, along with all queries this should have triggered at the next round, and thus is able to skip round $i^\ast+1$ at the price of an additional (at most) $q2^{\frac{q}{k}}$ queries.
\end{proof}

\begin{remark}
Note that in the above proof, while one can assume without loss of generality that the algorithm always makes exactly $q$ queries, one \emph{cannot} however assume that for any two such paths $(v_0,v_1,\dots,v_k,v^\ast)$ and $(u_0,u_1,\dots,u_k,u^\ast)$, $\abs{Q_{v_j}} = \abs{Q_{u_j}}$ for all $0\leq j\leq k$. That is, the number of queries made in round $j$ may not be the same depending on the path followed down by the algorithm, but instead depend adaptively on the previous queries made.
\end{remark}

The above remark shows the difficulty in extending the proof of~\autoref{prop:round:reduction:babystep} further than a single round. If one is willing to assume that the number of queries at each round is non-adaptive, it becomes possible to obtain a more general statement for $0 \leq \ell < k$; however, it is unclear how to proceed without this extra assumption, leading to the following question:
\begin{openquestion}\label{question:round:reduction:arbitrarystep}
  Can one obtain a general round-reduction upper bound for $0\leq \ell < k$ of the form $q_\ell \leq \phi(q_k,\ell,k)$, improving on~\autoref{fact:round:reduction:alltheway} for $\ell > 0$?
\end{openquestion}


\subsection{On the Connection with Communication Complexity}\label{ssec:misc:communication:complexity}

As exemplified in the proof of~\autoref{lemma:dt:lb}, there exists a striking parallel between the notion of $k$-round-adaptive testing algorithms, and that of $k$-round protocols in communication complexity. In this section, we make this parallel rigorous, and give a blackbox reduction between the two that one can leverage to establish lower bounds on $k$-round-adaptive testing.

In more detail, we build on the communication complexity methodology for proving property testing lower bounds due to~\cite{BBM:12} (more precisely, to the general formulation of this methodology as laid out in~\cite{Goldreich:13}). Although the results stated there hold for non-adaptive lower bounds (in the case of one-way communication or simultaneous message passing) or
 fully adaptive lower bounds in property testing (in the case of two-way communication), it is easy to obtain their counterpart for $k$-round-adaptive, given in~\autoref{theo:cc:lb:pt} below. But first, we need to recall some notations.
 
 In what follows, for a property $\property$, integer $k$, and parameters $\eps,\delta\in[0,1]$, we write $Q^{(k)}_\delta(\eps,\property)$ for the minimum query complexity of any $k$-round-adaptive tester for $\property$ with error probability $\delta$ and distance parameter $\eps$. Given a communication complexity predicate $F$, we let $\operatorname{CC}^{(k)}_{\delta}(F)$, $\overrightarrow{\operatorname{CC}}_{\delta}(F)$, and $\overleftarrow{\operatorname{CC}}_{\delta}(F)$ denote respectively the minimum communication complexity of a public-coin protocol for $F$ with error $\delta$ in (i) $k$-rounds, (ii) one-way from Alice to Bob, and (iii) one-way from Bob to Alice, respectively (note that the case $\delta=0$ then corresponds to protocols with perfect completeness).
 
 \begin{theorem}\label{theo:cc:lb:pt}
  Let $\Psi=(P,S)$ be a promise problem such that $P,S\subseteq\bool^{2n}$, $\property\subseteq\bool^\ell$ be a property, and $\eps,\delta>0$.
 Suppose the mapping $F\colon\bool^{2n}\to\bool^\ell$ satisfies the following two conditions:
  \begin{enumerate}[(i)]
    \item for every $(x,y)\in P\cap S$, it holds that $F(x, y)\in\property$;
    \item for every $(x,y)\in P\setminus S$, it holds that $F(x, y)$ is $\eps$-far from $\property$.
  \end{enumerate}
  Then $Q^{(k)}_\delta(\eps,\property) \geq \frac{1}{B+1}\operatorname{CC}^{(k+2)}_{2\delta}(\Psi)$, where $B\eqdef \max_{i\in[\ell]} \max( \overrightarrow{\operatorname{CC}}_{\frac{\delta}{n}}(F_i), \overleftarrow{\operatorname{CC}}_{\frac{\delta}{n}}(F_i) )$ (and $F_i(x,y)$ is the $i$'th bit of $F(x,y)$). 
  Moreover, if $B' \eqdef \max_{i\in[\ell]} \max( \overrightarrow{\operatorname{CC}}_{0}(F_i), \overleftarrow{\operatorname{CC}}_{0}(F_i) )$, then $Q^{(k)}_\delta(\eps,\property) \geq \frac{1}{B'+1}\operatorname{CC}^{(k+2)}_{\delta}(\Psi)$.
 \end{theorem}
 \begin{proof}
 The proof will be identical to that of~\cite[Theorem 3.1]{Goldreich:13}, where we only need to check that Alice and Bob can each simulate the execution of the property testing algorithm (using their public random coins), answering the queries made to $F(x,y)$ while preserving the number of rounds. Running the testing algorithm, Alice first sends the bits allowing Bob to compute the answers to the first $q_0$ queries, using her input $x$ and the one-way protocols for the relevant $F_i$'s. Bob then answers with the $q_0$ bits corresponding to the answers he computed, as well as the bits allowing Alice to compute the answers to the next $q_1$ queries made by the tester, using now his input $y$ and the one-way protocols for the relevant $F_i$'s. They do so for $k+1$ rounds of communication in total, until the last player to receive a message gets from the other player both the answers to the queries in $Q_{k-1}$ as well as the bits needed to compute (given their own input) the answers to the last $q_k$ queries. At that point, it only remains to use a last round of communication (the $(k+2)$'nd) to communicate to the other player the answers to these last $q_k$ queries, so that both Alice and Bob can finish running their copy of the testing algorithm and know the answer.
 
 Note that the number of bits communicated at round $1\leq i\leq k+2$ is by definition of $B$ (resp. $B'$) at most $B\cdot q_{i-1}+ q_{i-2}$ (resp. $B'\cdot q_{i-1}+ q_{i-2}$), so that at most $(B+1)q$ (resp. $(B'+1)q$) bits are communicated in total. This concludes the proof.
 \end{proof}
 
 To illustrate the above methodology, we show how it can be leveraged to prove a hierarchy of lower bounds on the power of $k$-adaptive testers for testing a very fundamental class of  Boolean functions, that of $m$-linear functions.\footnote{We observe that establishing the upper bound counterpart to this result would provide an answer to~\autoref{openquestion:single:property}, although one rather weak quantitatively. It also, as a special case, would separate adaptive and non-adaptive testing of $m$-linearity for $m=o(n)$, a longstanding open question~\cite{BK:12,BCK:14}.}
 \begin{proposition}
 Let $\mathsf{PAR}_s^n\subseteq 2^{2^n}$ denote the class of parities of size $s$ (over $n$ variables), and fix $m\eqdef \frac{\sqrt{n}}{2}$. Then, for any $0\leq k\leq \log^\ast m -2$, any $(k,q)$-round-adaptive tester for $\mathsf{PAR}_{2m}^n$ must satisfy $q=\bigOmega{m\log^{(k+2)}m}$.
 \end{proposition}
 \begin{proof}
 We will rely on a result of Sa\u{g}lam and Tardos~\cite{SaglamT:13}, which implies the following (tight) lower bound on the communication complexity of sparse set-disjointness ($\mathsf{DISJ}_{m}^{n}$, where both inputs $x,y\in\bitset^n$ are promised to have Hamming weight $m$):
\begin{theorem*}[{Corollary of \cite[Theorem 4]{SaglamT:13}}]
For any $1 \leq k \leq \log^\ast m$, any $k$-round probabilistic protocol for $\mathsf{DISJ}_{m}^{4m^2}$ with error probability at most $1/3$ must have communication $\bigOmega{m\log^{(k)} m}$.
\end{theorem*}
It then suffices to provide a reduction from $\mathsf{DISJ}_{m}^{4m^2}$ to testing $\mathsf{PAR}_{2m}^{4m^2}$. We follow the known reduction, as can be found in~\cite{BBM:12,BGMW:13}. Namely, on input $x\in\bitset^n$ (resp. $y\in\bitset^n$), Alice (resp. Bob) forms the parity function $\chi_x$ (resp. $\chi_y$). As
$\abs{x\oplus y} = \abs{x}+\abs{y}-2\abs{x\cap y} = 2m-2\abs{x\cap y}$, the function $\chi_{x\oplus y}$ is a $2(m-\abs{x\cap y})$-parity. Moreover, as for any $z\in\bitset^n$ we have $\chi_{x\oplus y}(z)=\chi_{x}(z)\oplus \chi_{y}(z)$, each query can be answered (with zero error) by one bit of communication in either direction.

Put in the language of our reduction theorem, $\Psi=(P,S)$ with $P=\setOfSuchThat{ u\in\bitset^n}{\abs{u}=m}^2$ and $S=\setOfSuchThat{(x,y)\in P}{ \abs{x\cap y} \neq 0 }$; while $\ell=2^n$, $\property=\textsf{PAR}_{2m}^n\subseteq 2^{\ell}$; and 
$F\colon \bitset^{2n}\to \bitset^\ell$ maps $(x,y)$ to the truth table of $\chi_{x\oplus y}$. Since any two distinct parities are at distance $\frac{1}{2}$, we can take any $\eps\leq \frac{1}{2}$. We then have $B'=1$, and by the theorem above we know that
$
	\operatorname{CC}^{(k+2)}_{1/3}(\Psi) = \bigOmega{m\log^{(k+2)} m}
$
for any $0\leq k \leq \log^\ast m - 2$. Invoking~\autoref{theo:cc:lb:pt} concludes the proof. 
 \end{proof}

\subsection{On the Relative Power of Round- and Tail-Adaptive Testers}\label{sec:misc:round:tail}

In this section, we show that the two notions of round- and tail-adaptive testers we introduced are not equivalent. As mentioned in~\autoref{sec:def}, while round-adaptive testers are at least as powerful as tail-adaptive ones, there exist properties for which the separation is strict:
\begin{theorem}
\label{theo:tailround:theorem}
  Fix any $\alpha \in (0,1)$. There exists a constant $\beta \in (0,1)$ such that, for every $n\in\N$, the following holds. 
  For every integer $0\leq k\leq n^{\beta}$, there exists a property $\property_k \subseteq \field{n}^{n^{1+\alpha}}$ such that, for any constant $\eps\in(0,1]$,
    \begin{enumerate}[(i)]
      \item there exists a $(k,\tildeO{k})$-round-adaptive (one-sided) tester for $\property_k$; yet
      \item any $(k,q)$-tail-adaptive (two-sided) tester for $\property_k$ must satisfy $q=\bigOmega{n}$.
    \end{enumerate}
\end{theorem}
\begin{proof}[Proof sketch]
The argument is very similar to that of~\autoref{theo:hierarchy:theorem}, and follows the same overall structure. Namely, we slightly modify the $k$-iterated function $f_k$ of~\autoref{sec:hierarchy} (which \emph{was} computable by a $(k,k+1)$-tail-adaptive algorithm) to rule out tail-adaptive algorithms but not round-adaptive ones:
that is, we define the function $f'_k\colon\field{n}^n \to \field{n}$ by
\begin{align*}
  f'_k(x) = \begin{cases}
  1 & \text{ if } x_{x,g_{k-1}(x)}=x_{x,g_{k-1}(x)+1 \bmod n} \\
  0 & \text{  otherwise.}
  \end{cases}
\end{align*}
(Perhaps more clearly, $f'_k$ is computed by iterating the pointer function $k$ times, and then checking if the value $x_i$ at the final coordinate $i\in[n]$ reached, and the value $x_{i+1}$ at the adjacent coordinate $i+1$, are equal.) It is not hard to see that the counterparts of~\autoref{claim:dt:ub} and~\autoref{lemma:dt:lb} still hold for $f'_k$: first, the function is still easy to compute by $(k,k+2)$-round-adaptive algorithms. However, because the very last round requires $2$ queries and not one (to query $x_i$ and $x_{i+1}$, once the value of $i=g_{k-1}(x)$ has been obtained), tail-round-adaptive algorithms are no longer able to leverage this, and analogously to~\autoref{lemma:dt:lb} we can conclude that there is no $(k,o(n/(k^2\log n)))$-round-adaptive (randomized) LDT algorithm which computes $f'_{k}$. It then only remains to lift this DT separation to property testing: we can do this as before (noting, in the case of lifting the lower bound, that the reduction of~\autoref{lemma:ldt:pt} preserves the number of queries per round, and thus the ``tailness'' of the algorithm).
\end{proof}

\section*{Acknowledgments}
We are grateful to Oded Goldreich for suggesting cycle freeness as a candidate natural property for proving an adaptivity hierarchy theorem, as well as for enlightening conversations that significantly contributed to this work; and wish to thank Rocco Servedio for helpful comments on an earlier version of this paper.

\clearpage
\bibliographystyle{alpha}
\bibliography{references} 

\clearpage
\appendix

\end{document}